\documentclass{llncs}
\usepackage{makeidx}  
\usepackage{tikz}
\usepackage{float}
\usepackage{amsmath}
\begin{document}
\frontmatter          
\pagestyle{headings}  

\mainmatter              
\title{Efficient AUC Optimization for Information Ranking Applications}
\titlerunning{AUC Optimization for Ranking}  
%
\author{Sean J. Welleck}
\authorrunning{Welleck} 
%
\tocauthor{Welleck, Sean}
\institute{IBM, USA\\
\email{swelleck@us.ibm.com}}
\maketitle              

\begin{abstract}
Adequate evaluation of an information retrieval system to estimate future performance is a crucial task. 
Area under the ROC curve (AUC) is widely used to evaluate the generalization of a retrieval system. 
However,
the objective function optimized in many retrieval systems is the error rate and not the AUC value.
This paper provides an efficient and effective non-linear approach to optimize AUC using additive regression trees, with a special emphasis on the use of multi-class AUC (MAUC) because multiple relevance levels are widely used in many ranking applications. 
Compared to a conventional linear approach, the performance of the non-linear approach is comparable on binary-relevance benchmark datasets and is better on  multi-relevance benchmark datasets. 

\keywords{machine learning, learning to rank, evaluation}
\end{abstract}

\section{Introduction}
In various information retrieval applications, a system may need to provide a ranking of candidate items that satisfies a criteria. 
For instance, a search engine must produce a list of results, ranked by their relevance to a user query. 
The relationship between items (e.g. documents) represented as feature vectors and their rankings (e.g. based on relevance scores) is often complex, so machine learning is used to learn a function that generates a ranking given a list of items.

The ranking system is evaluated using metrics that reflect certain goals for the system. 
The choice of metric, as well as its relative importance, varies by application area. 
For instance, a search engine may evaluate its ranking system with Normalized Discounted Cumulative Gain (NDCG), while a question-answering system evaluates its ranking using precision at 3; a high NDCG score is meant to indicate results that are relevant to a user's query, while a high precision shows that a favorable amount of correct answers were ranked highly. 
Other common metrics include Recall @ k, Mean Average Precision (MAP), and Area Under the ROC Curve (AUC). 

Ranking algorithms may optimize error rate as a proxy for improving metrics such as AUC, or may optimize the metrics directly. 
However, typical metrics such as NDCG and AUC are either flat everywhere or non-differentiable with respect to model parameters, making direct optimization with gradient descent difficult.

LambdaMART\cite{Burges2010f} is a ranking algorithm that is able to avoid this issue and directly optimize non-smooth metrics. 
It uses a gradient-boosted tree model and forms an approximation to the gradient whose value is derived from the evaluation metric. 
LambdaMART has been empirically shown to find a local optimum of NDCG, Mean Reciprocal Rank, and Mean Average Precision \cite{Donmez2009o}. 
An additional attractive property of LambdaMART is that the evaluation metric that LambdaMART optimizes is easily changed; the algorithm can therefore be adjusted for a given application area. 
This flexibility makes the algorithm a good candidate for a production system for general ranking, as using a single algorithm for multiple applications can reduce overall system complexity.

However, to our knowledge LambdaMART's ability to optimize AUC has not been explored and empirically verified in the literature. 
In this paper, we propose extensions to LambdaMART to optimize AUC and multi-class AUC, and show that the extensions can be computed efficiently. 
To evaluate the system, we conduct experiments on several binary-class and multi-class benchmark datasets. 
We find that LambdaMART with the AUC extension performs similarly to an SVM baseline on binary-class datasets, and LambdaMART with the multi-class AUC extension outperforms the SVM baseline on multi-class datasets.

\section{Related Work}

This work relates to two areas: LambdaMART and AUC optimization in ranking. LambdaMART was originally proposed in \cite{Wu2010a} and is overviewed in \cite{Burges2010f}. 
The LambdaRank algorithm, upon which LambdaMART is based, was shown to find a locally optimal model for the IR metrics NDCG@10, mean NDCG, MAP, and MRR \cite{Donmez2009o}. 
Svore \textit{et. al} \cite{Svore2011l} propose a modification to LambdaMART that allows for simultaneous optimization of NDCG and a measure based on click-through rate. 

Various approaches have been developed for optimizing AUC in binary-class settings.
Cortes and Mohri \cite{Cortes2004a} show that minimum error rate training may be insufficient for optimizing AUC, and demonstrate that the RankBoost algorithm globally optimizes AUC. 
Calders and Jaroszewicz \cite{Calders2007k} propose a smooth polynomial approximation of AUC that can be optimized with a gradient descent method. 
Joachims \cite{Joachims2005a} proposes an SVM method for various IR measures including AUC, and evaluates the system on text classification datasets. 
The SVM method is used as the comparison baseline in this paper.

\section{Ranking Metrics}
We will first provide a review of the metrics used in this paper. 
Using document retrieval as an example, consider $n$ queries $Q_{1}...Q_{n}$, and let $n(i)$ denote the number of documents in query $Q_i$. 
Let $d_{ij}$ denote document $j$ in query $Q_i$, where $i \in {1, ..., n}$, $j \in {1 ... n(i)}$. 

\subsection{Contingency Table Metrics}
Several IR metrics are derived from a model's contingency table, which contains the four entries True Positive (TP), False Positive (FP), False Negative (FN), and True Negative (TN):
\begin{center}
\begin{tabular}{|c|c|c|}
\hline                        & $y = \ell_p$ & $y = \ell_n$ \\
\hline $f(x) = \ell_p$ & TP                & FP                \\ 
\hline $f(x) = \ell_n$ & FN                & TN                \\ 
\hline 
\end{tabular} 
\end{center}
where $y$ denotes an example's label, $f(x)$ denotes the predicted label, $\ell_p$ denotes the class label considered positive, and $\ell_n$ denotes the class label considered negative. 

Measuring the precision of the first $k$ ranked documents is often important in ranking applications. 
For instance, $Precision@1$ is important for question answering systems to evaluate whether the system's top ranked item is a correct answer. 
Although precision is a metric for binary class labels, many ranking applications and standard datasets have multiple class labels. 
To evaluate precision in the multi-class context we use Micro-averaged Precision and Macro-averaged Precision, which summarize precision performance on multiple classes \cite{Manning2008I}. 

\subsubsection{Micro-averaged Precision}
Micro-averaged Precision pools the contingency tables across classes, then computes precision using the pooled values:
\begin{equation}
Precision_{micro}=\frac{\sum_{c=1}^{C}TP_c}{\sum_{c=1}^{C}TP_c+FP_c}
\end{equation} 
where $C$ denotes the number of classes, $TP_c$ is the number of true positives for class $c$, and $FP_c$ is the number of false positives for class $c$. 

$Precision_{micro}@k$ is measured by using only the first $k$ ranked documents in each query:
\begin{equation}
Precision_{micro}@k=\frac{1}{n}\sum_{i=1}^n\frac{ \sum_{c=1}^C \left| \lbrace d_{ij} | y_{j}=c, j\in \lbrace 1,..., k\rbrace \rbrace \right|}{(C)(k)}
\end{equation}
Micro-averaged precision indicates performance on prevalent classes, since prevalent classes will contribute the most to the $TP$ and $FP$ sums.

\subsubsection{Macro-averaged Precision}
Macro-averaged Precision is a simple average of per-class precision values:
\begin{equation}
Precision_{macro}=\frac{1}{C}\sum_{c=1}^C\frac{TP_c}{TP_c+FP_c}
\end{equation} 
Restricting each query's ranked list to the first $k$ documents gives:
\begin{equation}
Precision_{macro}@k=\frac{1}{C}\sum_{c=1}^C\sum_{i=1}^n\frac{ \left| \lbrace d_{ij} | y_{j}=c, j\in \lbrace 1,..., k\rbrace \rbrace \right|}{k}
\end{equation} 
Macro-averaged precision indicates performance across all classes regardless of prevalence, since each class's precision value is given equal weight.

\subsubsection{AUC}
AUC refers to the area under the ROC curve. The ROC curve plots True Positive Rate $(TPR = \frac{TP}{TP + FN}) $ versus False Positive Rate $(FPR = \frac{FP}{FP + TN})$, with $TPR$ appearing on the y-axis, and $FPR$ appearing on the x-axis.

Each point on the ROC curve corresponds to a contingency table for a given model. 
In the ranking context, the contingency table is for the ranking cutoff $k$; the curve shows the $TPR$ and $FPR$ as $k$ changes. 
A model is considered to have better performance as its ROC curve shifts towards the upper left quadrant. 
The AUC measures the area under this curve, providing a single metric that summarizes a model's ROC curve and allowing for easy comparison.

We also note that the AUC is equivalent to the Wilcoxon-Mann-Whitney statistic \cite{Cortes2004a} and can therefore be computed using the number of correctly ordered document pairs. 
Fawcett \cite{Fawcett2006a} provides an efficient algorithm for computing AUC.

\subsubsection{Multi-Class AUC}

The standard AUC formulation is defined for binary classification. 
To evaluate a model using AUC on a dataset with multiple class labels, AUC can be extended to multi-class AUC (MAUC). 

We define the \textit{class reference} AUC value $AUC(c_i)$ as the AUC when class label $c_i$ is viewed as positive and all other labels as negative. The multi-class AUC is then the weighted sum of class reference AUC values, where each class reference AUC is weighted by the proportion of the dataset examples with that class label, denoted $p(c_i)$ \cite{Fawcett2006a}:
\begin{equation} \label{eq:mauc}
MAUC=\sum_{i=1}^C AUC(c_i)*p(c_i)\ .
\end{equation}
Note that the class-reference AUC of a prevalent class will therefore impact the MAUC score more than the class-reference AUC of a rare class.

\section {$\lambda$-Gradient Optimization of the MAUC function}

We briefly describe LambdaMART's optimization procedure here and refer the reader to \cite{Burges2010f} for a more extensive treatment.
LambdaMART uses a gradient descent optimization procedure
that only requires the gradient, rather than the objective function, to be defined. 
The objective function can in principal be left undefined, since only the gradient is required to perform gradient descent. 
Each gradient approximation, known as a $\lambda$-gradient, focuses on document pairs $(d_i, d_j)$ of conflicting relevance values (document $d_i$ more or less relevant than document $d_j$):

\begin{equation}
	\lambda_i = \sum_{j\in (d_i, d_j) \mid \ell_i \neq \ell_j}\lambda_{ij}\qquad \qquad
	\lambda_{ij} = S_{ij} \left| \Delta M_{IR_{ij}}	\frac{ \partial{C_{ij} } } { \partial{o_{ij}  } }  \right|
\end{equation}
with $S_{ij} = 1$ when $l_i > l_j$ and $-1$ when $l_j < l_i$.

The $\lambda$-gradient includes the change in IR metric, $\Delta M_{IR_{ij}}$,
from swapping the rank positions of the two documents, discounted by a function of the score difference between the documents. 

For a given sorted order of the documents, 
the objective function is simply a weighted version of the RankNet \cite{Burges2005RankNet} cost function. 
The RankNet cost is a pairwise cross-entropy cost applied to the logistic of the difference of the model scores. 
If document $d_i$, with score $s_i$, is to be ranked higher than document $d_j$, with score $s_j$, then the RankNet cost can be written as follows:
\begin{equation}
C(o_{ij}) = o_{ij} + \log( 1 + e^{o_{ij}} ) 
\end{equation}
where $o_{ij} = s_j - s_i$ is the score difference of a pair of documents in a query.
The derivative of the RankNet cost according to the difference in score is
\begin{equation}
\frac{ \partial{C_{ij} } } { \partial{o_{ij}  } } = \frac{1}{ (1+ e^{o_{ij}}) }\ .
\end{equation}

The optimization procedure using $\lambda$-gradients was originally defined using $ \Delta NDCG $ as the $\Delta M_{IR}$ term in order to optimize NDCG. 
$ \Delta MAP $ and $ \Delta MRR $ were also used to define effective $\lambda$-gradients for MAP and MRR, respectively. 
In this work, we adopt the approach of replacing the $\Delta M_{IR}$ term to define $ \lambda $-gradients for AUC and multi-class AUC.

\subsection{$\lambda $-gradients for AUC and multi-class AUC}
\subsubsection{$\lambda$-$AUC$}
Defining the $\lambda$-gradient for AUC requires deriving a formula for $\Delta AUC_{ij}$ that can be efficiently computed. 
Efficiency is important since in every iteration, the term is computed for $O(n(i)^2)$ document pairs for each query $Q_i$. 

To derive the $\Delta AUC_{ij}$ term, we begin with the fact that AUC is equivalent to the Wilcoxon-Mann-Whitney statistic \cite{Cortes2004a}. 
For documents $d_{p_1},...,d_{p_m}$ with positive labels and documents $d_{n_1}, ..., d_{n_n}$ with negative labels, we have:

\begin{equation}
AUC=\frac{\sum_{i=1}^m\sum_{j=1}^n I(f(d_{p_i})>f(d_{n_j}))}{mn}\ .
\end{equation}

The indicator function $I$ is $1$ when the ranker assigns a score to a document with a positive label that is higher than the score assigned to a document with a negative label. 
Hence the numerator is the number of correctly ordered pairs, and we can write \cite{Joachims2005a}:

\begin{equation}
AUC = \frac{CorrectPairs}{mn}
\end{equation}
where
\begin{equation}
CorrectPairs=\left| \lbrace(i, j) : (\ell_i > \ell_j)\ and\ (f(d_i) > f(d_j))\rbrace\right|\ .
\end{equation}
Note that a pair with equal labels is \textit{not} considered a correct pair, since a document pair $(d_i, d_j)$ contributes to $CorrectPairs$ if and only if $d_i$ is ranked higher than $d_j$ in the ranked list induced by the current model scores, and $\ell_i > \ell_j$. 

We now derive a formula for computing the $\Delta AUC_{ij}$ term in $O(1)$ time, given the ranked list and labels. 
This avoids the brute-force approach of counting the number of correct pairs before and after the swap, in turn providing an efficient way to compute a $\lambda $-gradient for $AUC$. 
Specifically, we have:

\begin{theorem}
Let $d_1,...,d_{m+n}$ be a list of documents with $m$ positive labels and $n$ negative labels, denoted $\ell_1,...,\ell_{m+n}$, with $\ell_i \in \lbrace 0, 1 \rbrace$. For each document pair $(d_i, d_j),\ i, j \in \lbrace1, ..., m+n\rbrace$,
\begin{equation}\label{eq:deltaauc}
\Delta AUC_{ij} = \frac{(\ell_j - \ell_i)(j - i)}{mn}\ .
\end{equation}
\end{theorem}

\begin{proof}
To derive this formula, we start with
\begin{equation}
\Delta AUC_{ij} = \frac{CP_{swap} - CP_{orig}}{mn}
\end{equation}
where $CP_{swap}$ is the value of $CorrectPairs$ after swapping the scores assigned to documents $i$ and $j$, and $CP_{orig}$ is the value of $CorrectPairs$ prior to the swap. 
Note that the swap corresponds to swapping positions of documents $i$ and $j$ in the ranked list.
The numerator of $\Delta AUC_{ij}$ is the change in the number of correct pairs due to the swap.  
The following lemma shows that we only need to compute the change in the number of correct pairs \textit{for the pairs of documents within the interval [i,j]} in the ranked list.

\begin{lemma}
Let $(d_a, d_b)$ be a document pair where at least one of $(a, b) \notin [i, j]$. Then after swapping documents $(d_i, d_j)$, the pair correctness of $(d_a, d_b)$ will be left unchanged or negated by another pair.
\end{lemma}
\begin{proof}
Without loss of generality, assume $a < b$. There are five cases to consider.

\textbf{case $a\notin [i, j]$, $b\notin [i, j]$}: Then the pair $(d_a, d_b)$ does not change due to the swap, therefore its pair correctness does not change.

Note that unless one of $a$ or $b$ is an endpoint $i$ or $j$, the pair $(d_a, d_b)$ does not change. Hence we now assume that one of $a$ or $b$ is an endpoint $i$ or $j$.

\textbf{case $a < i$, $b = i$}: The pair correctness of $(d_a, d_b)$ will change if and only if $\ell_a = 1,\ \ell_b=1,\ \ell_j=0$ prior to the swap. 
But then the pair correctness of $(d_i, d_j)$ will change from correct to not correct, canceling out the change (see Fig. \ref{figure:swap}).

\begin{figure}
\centering
\begin{tikzpicture}
\draw (0,0) -- (0,4);
\draw (1,0) -- (1,4);

\draw (2.5,0) -- (2.5,4);
\draw (3.5,0) -- (3.5,4);

\draw[loosely dotted] (0, 1.25) -- (1, 1.25);
\draw[loosely dotted] (0, 2.75) -- (1, 2.75);

\draw[loosely dotted] (2.5, 1.25) -- (3.5, 1.25);
\draw[loosely dotted] (2.5, 2.75) -- (3.5, 2.75);

\node[draw=none] at (- 0.25, 2.75) {i};
\node[draw=none] at (- 0.25, 1.25) {j};

\node[draw=none] at (2.25, 2.75) {j};
\node[draw=none] at (2.25, 1.25) {i};

\node[draw=none] at (1.25, 3.5) {a};
\node[draw=none] at (1.25, 2.75) {b};

\node[draw=none] at (3.75, 3.5) {a};
\node[draw=none] at (3.75, 2.75) {b};

\node[draw=none] at (0.5, 3.5) {1};   
\node[draw=none] at (0.5, 2.75) {1}; 
\node[draw=none] at (0.5, 1.25) {0}; 

\node[draw=none] at (3.0, 3.5) {1};   
\node[draw=none] at (3.0, 2.75) {0}; 
\node[draw=none] at (3.0, 1.25) {1}; 

\draw[->] (1.25, 2.0) -- (2.25, 2.0);
\node[draw=none] at (1.75, 2.15) {swap};
\end{tikzpicture}
\caption{Document swap for case $a < i$, $b=i$, with $\ell_a = 1,\ \ell_b=1,\ \ell_j=0$}\label{figure:swap}
\end{figure}
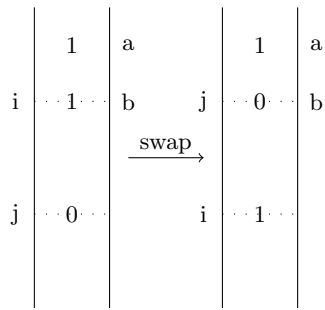
\textbf{case $a < i$, $b = j$}: Then the pair correctness of $(d_a, d_b)$ will change if and only if $\ell_a = 1,\ \ell_b=1,\ \ell_i=0$ prior to the swap. 
But then the pair correctness of $(d_a, d_i)$ will change from correct to not correct, canceling out the change.

\textbf{case $a = i$, $b > j$}: Then pair correctness of $(d_a, d_b)$ will change if and only if $\ell_a = 0,\ \ell_b=0,\ \ell_j=1$ prior to the swap. 
But then the pair correctness of $(d_j, d_b)$ will change from correct to not correct, canceling out the change.

\textbf{case $a = j$, $b > j$}: Then pair correctness of $(d_a, d_b)$ will change if and only if $\ell_a = 0,\ \ell_b=0,\ \ell_i=1$ prior to the swap. 
But then the pair correctness of $(d_i, d_b)$ will change from correct to not correct, canceling out the change.

Hence in all cases, either the pair correctness stays the same, or the pair $(d_a, d_b)$ changes from not correct to correct and an additional pair changes from correct to not correct, thus canceling out the change with respect to the total number of correct pairs after the swap. \qed
\end{proof}
Lemma 1 shows that the difference in correct pairs $CP_{swap} - CP_{orig}$ is equivalent to 
$CP_{swap_{[i, j]}} - CP_{orig_{[i, j]}}$, namely the change in the number of correct pairs within the interval [i,j]. 
Lemma 2 tells us that this value is simply the length of the interval [i,j].

\begin{lemma}
Assume $i < j$. Then
\begin{equation}
CP_{swap_{[i, j]}} - CP_{orig_{[i, j]}} = (\ell_j - \ell_i) (j - i)\ .
\end{equation}
\end{lemma}
\begin{proof}
There are three cases to consider.

\textbf{case $\ell_i=\ell_j$}: The number of correct pairs will not change since no document labels change due to the swap. 
Hence $CP_{swap_{[i, j]}} - CP_{orig_{[i, j]}} = 0 = (\ell_j - \ell_i) (j - i)$.

\textbf{case $\ell_i=1$, $\ell_j=0$}: Before swapping, each pair $(i, k)$, $i<k\leq j$ such that $\ell_k=0$ is a correct pair. 
After the swap, each of these pairs is not a correct pair. 
There are $n_{l_0[i,j]}$ such pairs, namely the number of documents in the interval $[i,j]$ with label $0$.

Each pair $(k, j)$, $i\leq k < j$ such that $\ell_k=1$ is a correct pair before swapping, and not correct after swapping. 
There are $n_{l_1[i,j]}$ such pairs, namely the number of documents in the interval $[i,j]$ with label $1$.

Every other pair remains unchanged, therefore 
\begin{equation}
n_{l_0[i,j]} + n_{l_1[i,j]}=j - i
\end{equation} 
pairs changed from correct to not correct, corresponding to a decrease in the number of correct pairs. 
Hence we have:
\begin{center}
$CP_{swap_{[i, j]}} - CP_{orig_{[i, j]}} = - (j - i) = (\ell_j - \ell_i) (j - i)$.
\end{center}

\textbf{case $\ell_i=0$, $\ell_j=1$}: Before swapping, each pair $(i, k)$, $i<k\leq j$ such that $\ell_k=0$ is not a correct pair. 
After the swap, each of these pairs is a correct pair. 
There are $n_{l_0[i,j]}$ such pairs, namely the number of documents in the interval $[i,j]$ with label $0$.

Each pair $(k, j)$, $i\leq k < j$ such that $\ell_k=1$ is not a correct pair before swapping, and is correct after swapping. 
There are $n_{l_1[i,j]}$ such pairs, namely the number of documents in the interval $[i,j]$ with label $1$.

Each pair $(i, k)$, $i<k\leq j$ such that $\ell_k=1$ remains not correct. 
Each pair $(k, j)$, $i\leq k< j$ such that $\ell_k=0$ remains not correct. 
Every other pair remains unchanged. Therefore 
\begin{equation}
n_{l_0[i,j]} + n_{l_1[i,j]}=j - i
\end{equation}
pairs changed from not correct to correct, corresponding to an increase in the number of correct pairs. 
Hence we have:

\begin{center}
$CP_{swap_{[i, j]}} - CP_{orig_{[i, j]}} = (j - i) = (\ell_j - \ell_i) (j - i)\ .$
\end{center} \qed
\end{proof}
Therefore by Lemmas 1 and 2, we have:
\begin{align*}
\Delta AUC_{ij} & = \frac{CP_{swap} - CP_{orig}}{mn} \\
                          & = \frac{CP_{swap_{[i,j]}} - CP_{orig_{[i,j]}}}{mn}  \\
                          & = \frac{(\ell_i - \ell_j)(j - i)}{mn}
\end{align*}
completing the proof of Theorem 1. \qed
\end{proof}
Applying the formula from Theorem 1 to the list of documents sorted by the current model scores, we define the $\lambda$-gradient for AUC as:
\begin{equation}
\lambda_{AUC_{ij}} =  S_{ij}\left|\Delta AUC_{ij}\frac{ \partial{C_{ij} } } { \partial{o_{ij}  }}\right|
\end{equation}
where $S_{ij}$ and $\frac{ \partial{C_{ij} } } { \partial{o_{ij}  } }$ are as defined previously, and $\Delta AUC_{ij} = \frac{(\ell_i - \ell_j)(j - i)}{mn}\ $.

\subsubsection{$\lambda$-$MAUC$}
To extend the $\lambda$-gradient for AUC to a multi-class setting, we consider the multi-class AUC definition found in equation \ref{eq:mauc}. 
Since MAUC is a linear combination of class-reference AUC values, to compute $\Delta MAUC_{ij}$ we can compute the change in each class-reference AUC value $\Delta AUC(c_k)$ separately using equation \ref{eq:deltaauc} and weight each $\Delta$ value by the proportion $p(c_k)$, giving:

\begin{equation}
\Delta MAUC_{ij}=\sum_{k=1}^{C} \Delta AUC(c_k)_{ij}*p(c_k)\ .
\end{equation}
Using this term and the previously defined terms $S_{ij}$ and $\frac{ \partial{C_{ij} } } { \partial{o_{ij}  }}$, we define the $\lambda $-gradient for $MAUC$ as:

\begin{equation}
\lambda_{MAUC_{ij}} = S_{ij}\left|\Delta MAUC_{ij}\frac{ \partial{C_{ij} } } { \partial{o_{ij}  }}\right|\ .
\end{equation}

\section {Experiments}

Experiments were conducted on binary-class datasets to compare the AUC performance of LambdaMART trained with the AUC $\lambda$-gradient, referred to as LambdaMART-AUC, against a baseline model. 
Similar experiments were conducted on multi-class datasets to compare LambdaMART trained with the MAUC $\lambda$-gradient, referred to as LambdaMART-MAUC, against a baseline in terms of MAUC. 
Differences in precision on the predicted rankings were also investigated.

The LambdaMART implementation used in the experiments was a modified version of the JForests learning to rank library \cite{Ganjisaffar2011b}. 
This library showed the best NDCG performance out of the available Java ranking libraries in preliminary experiments. 
We then implemented extensions required to compute the AUC and multi-class AUC $\lambda $-gradients. For parameter tuning, a learning rate was chosen for each dataset by searching over the values $\lbrace 0.1, 0.25, 0.5, 0.9 \rbrace$ and choosing the value that resulted in the best performance on a validation set. 

As the comparison baseline, we used a Support Vector Machine (SVM) formulated for optimizing AUC. 
The SVM implementation was provided by the SVM-Perf \cite{Joachims2005a} library.
The ROCArea loss function was used, and the regularization parameter $c$ was chosen by searching over the values $\lbrace 0.1, 1, 10, 100 \rbrace$ and choosing the value that resulted in the best performance on a validation set. 
For the multi-class setting, a binary classifier was trained for each individual relevance class. 
Prediction scores for a document $d$ were then generated by computing the quantity $\sum^{C}_{c=1}cf_{c}(d)$, where $C$ denotes the number of classes, and $f_c$ denotes the binary classifier for relevance class $c$. 
These scores were used to induce a ranking of documents for each query.

\subsection{Datasets}
For evaluating LambdaMART-AUC, we used six binary-class web-search datasets from the LETOR 3.0 \cite{Qin2010l} Gov dataset collection, named td2003, td2004, np2003, np2004, hp2003, and hp2004. 
Each dataset is divided into five folds and contains feature vectors representing query-document pairs and binary relevance labels.

For evaluating LambdaMART-MAUC, we used four multi-class web-search datasets:
versions 1.0 and 2.0 of the Yahoo! Learning to Rank Challenge \cite{Chapelle2011y} dataset, and the mq2007 and mq2008 datasets from the LETOR 4.0 \cite{Letor4} collection. 
The Yahoo! and LETOR datasets are divided into two and five folds, respectively. 
Each Yahoo! dataset has integer relevance scores ranging from 0 (not relevant) to 4 (very relevant), while the LETOR datasets have integer relevance scores ranging from 0 to 2. 
The LETOR datasets have 1700 and 800 queries, respectively, while the larger Yahoo! datasets have approximately 20,000 queries.

\subsection{Results}
\subsubsection{AUC}
On the binary-class datasets, LambdaMART-AUC and SVM-Perf performed similarly in terms of AUC and Mean-Average Precision. 
The results did not definitively show that either algorithm was superior on all datasets; LambdaMART-AUC had higher AUC scores on 2 datasets (td2003 and td2004), lower AUC scores on 3 datasets (hp2003, hp2004, np2004), and a similar score on np2003.
In terms of MAP, LambdaMART-AUC was higher on 2 datasets (td2003 and td2004), lower on 2 datasets (np2004, hp2004), and similar on 2 datasets (np2003, hp2003).
The results confirm that LambdaMART-AUC is an effective option for optimizing AUC on binary datasets, since the SVM model has previously been shown to perform effectively.

\subsubsection{MAUC}
Table \ref{table:mauc} shows the MAUC scores on held out test sets for the four multi-class datasets.
The reported value is the average MAUC across all dataset folds. 
The results indicate that in terms of optimizing Multi-class AUC, LambdaMART-MAUC is as effective as SVM-Perf on the LETOR datasets, and more effective on the larger Yahoo! datasets.
\begin{table*}
\centering
\caption{Summary of Multi-class AUC on test folds}\label{table:mauc}
\begin{tabular}{|c|c|c|c|c|}
\hline 
\rule[-1ex]{0pt}{3.5ex}   & Yahoo V1 & Yahoo V2 & mq2007 & mq2008 \\ 
\hline 
\rule[-1ex]{0pt}{2.5ex} LambdaMART-MAUC     & \textbf{0.594} & \textbf{0.592} & \textbf{0.662} & 0.734\\ 
\hline 
\rule[-1ex]{0pt}{2.5ex} SVM-Perf     & 0.576 & 0.576 & 0.659 & \textbf{0.737}\\ 
\hline 
\end{tabular}
\end{table*}
\begin{table*}
\centering
\caption{Summary of Mean Average Precision on test folds}\label{table:map}
\begin{tabular}{|c|c|c|c|c|}
\hline 
\rule[-1ex]{0pt}{3.5ex}   & Yahoo V1 & Yahoo V2 & mq2007 & mq2008 \\ 
\hline 
\rule[-1ex]{0pt}{2.5ex} LambdaMART-MAUC     & \textbf{0.862} & \textbf{0.858} & \textbf{0.466} & \textbf{0.474}\\ 
\hline 
\rule[-1ex]{0pt}{2.5ex} SVM-Perf     & 0.837 & 0.837 & 0.450 & 0.458\\ 
\hline 
\end{tabular} 
\end{table*}

Additionally, the experiments found that LambdaMART-MAUC outperformed SVM-Perf in terms of precision in all cases. 
Table \ref{table:map} shows the Mean Average Precision scores for the four datasets. LambdaMART-MAUC also had higher $Precision_{micro}@k$ and $Precision_{macro}@k$ on all datasets, for $k=1,... ,10$. 
For instance, Figure \ref{fig:plots} shows the values of $Precision_{micro}@k$ and $Precision_{macro}@k$ for the Yahoo! V1 dataset.

The class-reference AUC scores indicate that LambdaMART-MAUC and SVM-Perf arrive at their MAUC scores in different ways. 
LambdaMART-MAUC focuses on the most prevalent class; each $\Delta AUC(c_i)$ term for a prevalent class receives a higher weighting than for a rare class due to the $p(c_i)$ term in the $\lambda_{MAUC}$ computation. 
As a result the $\lambda$-gradients in LambdaMART-MAUC place more emphasis on achieving a high $AUC(c_1)$ than a high $AUC(c_4)$. 
Table \ref{table:aucs} shows the class-reference AUC scores for the Yahoo! V1 dataset. 
We observe that LambdaMART-MAUC produces better $AUC(c_1)$ than SVM-Perf, but worse $AUC(c_4)$, since class 1 is much more prevalent than class 4; 48\% of the documents in the training set with a positive label have a label of class 1, while only 2.5\% have a label of class 4. 

Finally, we note that on the large-scale Microsoft Learning to Rank Dataset MSLR-WEB10k \cite{Mslrdataset}, the SVM-Perf training failed to converge on a single fold after 12 hours. 
Therefore training a model for each class for every fold was impractical using SVM-Perf, while LambdaMART-MAUC was able to train on all five folds in less than 5 hours. 
This further suggests that LambdaMART-MAUC is preferable to SVM-Perf for optimizing MAUC on large ranking datasets.\\[1.2in]

\begin{figure}
\centering
\includegraphics[natwidth=1152,width=170px]{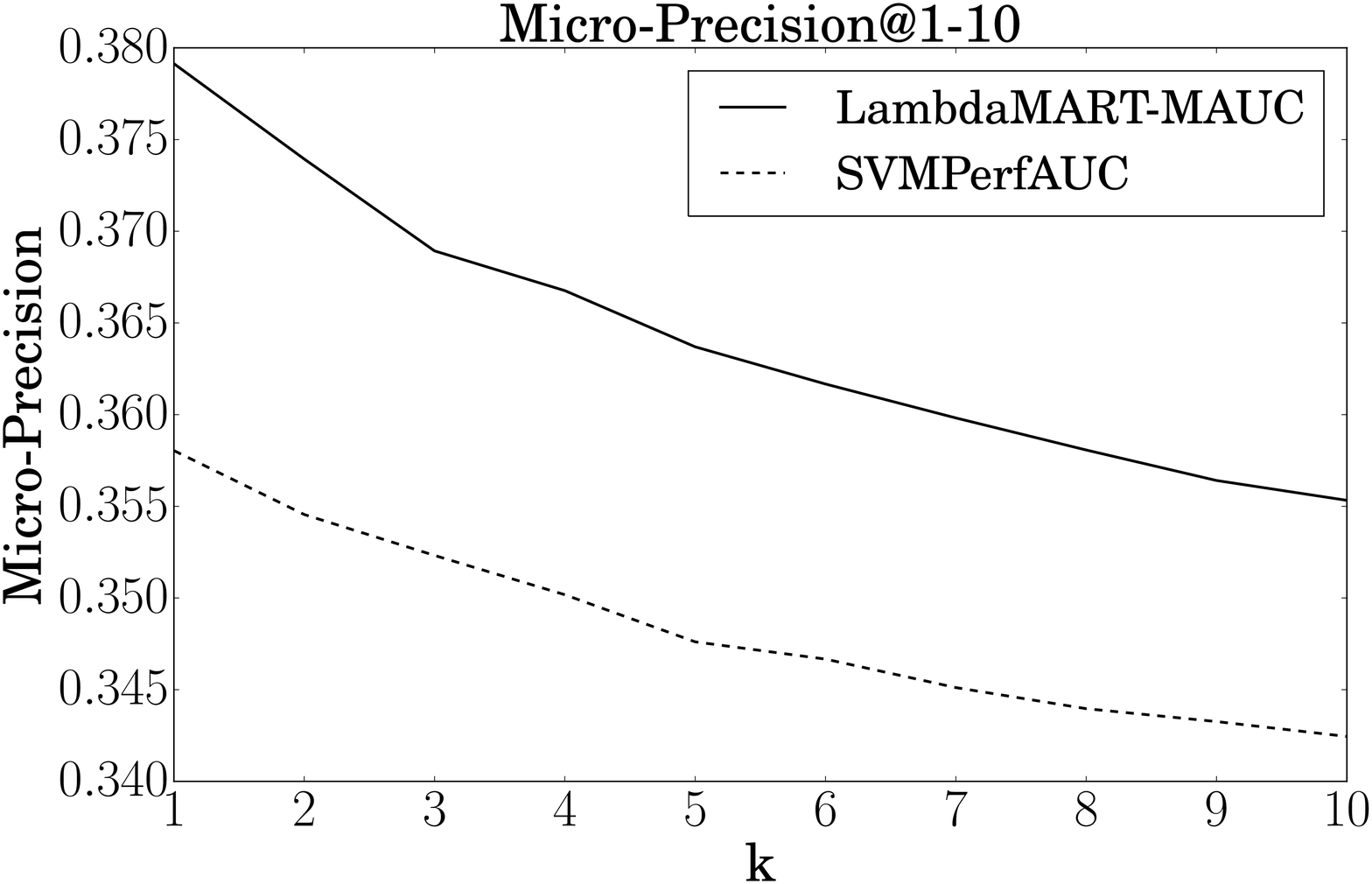}
\includegraphics[natwidth=1152,width=170px]{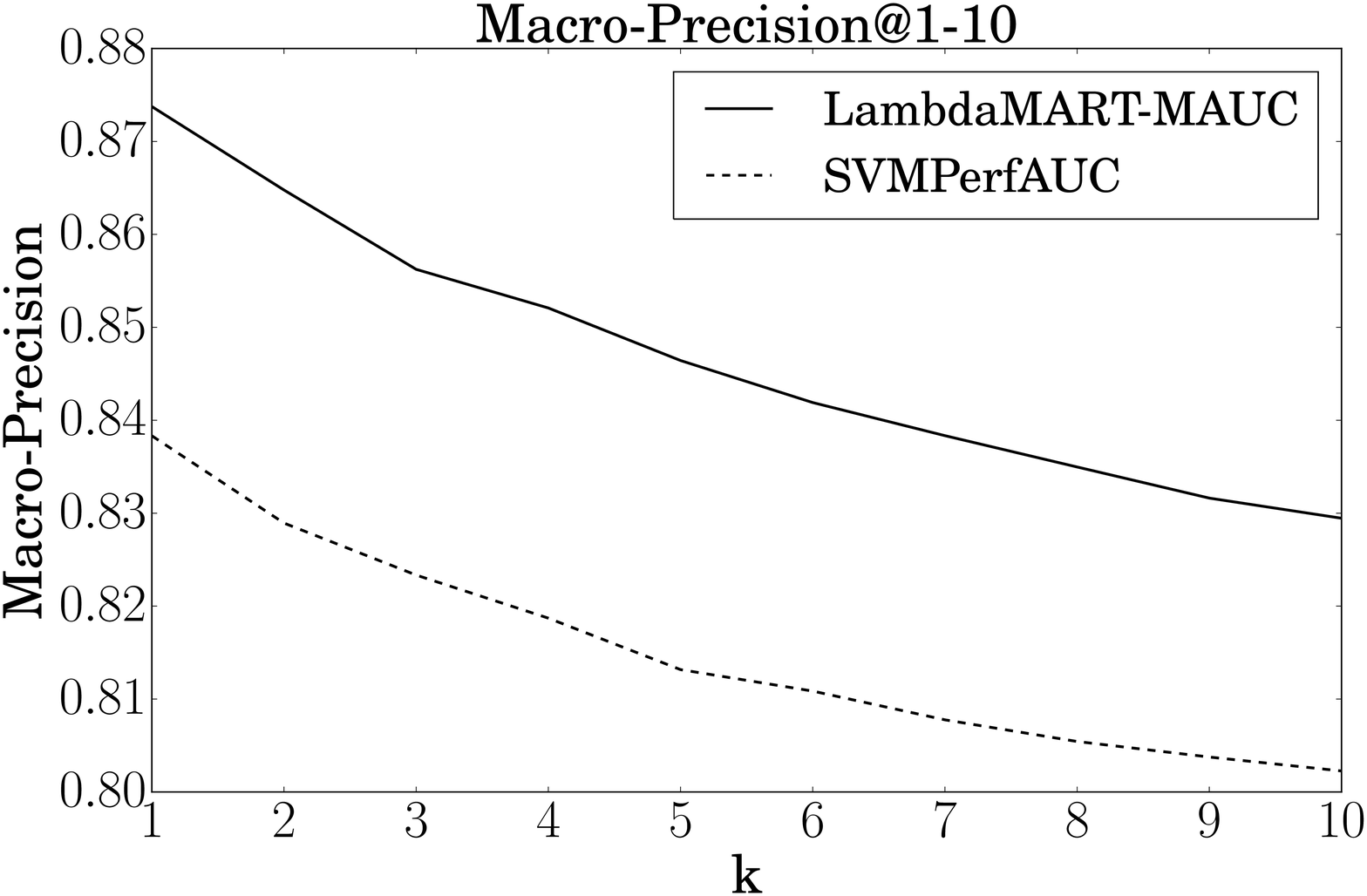} 
\caption{Micro and Macro Precision@1-10 on the Yahoo! V1 test folds}
\label{fig:plots}
\end{figure}

\begin{table}
\centering
\caption{Summary of class-reference AUC scores on the Yahoo! V1 test folds}\label{table:aucs}
\begin{tabular}{|c|c|c|c|c|}
\hline 
\rule[-1ex]{0pt}{3.5ex}   & $AUC_1$ & $AUC_2$ & $AUC_3$ & $AUC_4$ \\ 
\hline 
\rule[-1ex]{0pt}{2.5ex} LambdaMART-MAUC     & \textbf{0.503} & \textbf{0.690} & 0.757 & 0.831\\ 
\hline 
\rule[-1ex]{0pt}{2.5ex} SVM-Perf     & 0.474 & 0.682 & \textbf{0.796} & \textbf{0.920}\\ 
\hline 
\end{tabular} 
\end{table}

\section{Conclusions}
We have introduced a method for optimizing AUC on ranking datasets using a gradient-boosting framework. 
Specifically, we have derived gradient approximations for optimizing AUC with LambdaMART in binary and multi-class settings, and shown that the gradients are efficient to compute. 
The experiments show that the method performs as well as, or better than, a baseline SVM method, and performs especially well on large, multi-class datasets. 
In addition to adding LambdaMART to the portfolio of algorithms that can be used to optimize AUC, our extensions expand the set of IR metrics for which LambdaMART can be used.

There are several possible future directions. 
One is investigating local optimality of the solution produced by LambdaMART-AUC using Monte Carlo methods. 
Other directions include exploring LambdaMART with multiple objective functions to optimize AUC, and creating an extension to optimize area under a Precision-Recall curve rather than an ROC curve.

\section*{Acknowledgements}
Thank you to Dwi Sianto Mansjur for giving helpful guidance and providing valuable comments about this paper.
%
%

\bibliographystyle{splncs03}
\bibliography{refs}

\end{document}